\documentclass[11pt,a4paper,english,reqno]{amsart}
\usepackage{amsmath,amssymb,amsfonts,epsfig,mathrsfs}
\usepackage[T1]{fontenc}

\usepackage{color}
\usepackage{array}
\usepackage{amsthm}
\usepackage{amstext}
\usepackage{graphicx}
\usepackage{setspace}
\usepackage[margin=2.5cm]{geometry}
\usepackage{bbm}
\usepackage{color}
\usepackage{enumitem}
\allowdisplaybreaks[4]

\usepackage{amscd,psfrag}
\usepackage{yhmath}
\usepackage[mathscr]{eucal}

\usepackage{slashed}

\makeatletter
\pdfpageheight\paperheight
\pdfpagewidth\paperwidth

\usepackage{epstopdf}
\usepackage{chngcntr}
\usepackage{mathrsfs}

\usepackage{indentfirst}	

\usepackage[normalem]{ulem}
\theoremstyle{plain}

\newtheorem{definition}{Definition}[section]
\newtheorem{theorem}[definition]{Theorem}
\newtheorem*{theorem*}{Theorem}

\newtheorem{assumption}[definition]{Assumption}

\newtheorem*{remark*}{Remark}
\newtheorem*{sideremark*}{Side Remark}\newtheorem*{mt*}{Main Theorem}

\newtheorem*{claim*}{Claim}
\newtheorem*{q*}{Question}

\newtheorem*{corollary*}{Corollary}
\newtheorem*{proposition*}{Proposition}

\newcommand{\R}{\mathbb{R}}

\newcommand{\na}{\nabla}

\newcommand{\p}{\partial}
\newcommand{\e}{\epsilon}

\newcommand{\map}{\rightarrow}

\newcommand{\two}{{\rm II}}

\newcommand{\OO}{{\mathcal{O}}}
\newcommand{\rt}{{\mathcal{R}_\theta}}
\newcommand{\ta}{{\widetilde{\aaa}}}

\allowdisplaybreaks[4]

\def\XXint#1#2#3{{\setbox0=\hbox{$#1{#2#3}{\int}$ }
\vcenter{\hbox{$#2#3$ }}\kern-.6\wd0}}

\newcommand{\FF}{{\bf F}}
\newcommand{\TT}{{\bf T}}
\newcommand{\WW}{{\bf W}}
\newcommand{\aaa}{{\bf \alpha}}
\newcommand{\gl}{{\mathfrak{gl}(3;\R)}}
\newcommand{\curl}{{{\rm curl}\,}}\newcommand{\dv}{{{\rm div}\,}}

\newcommand{\n}{{\bf n}}

\newcommand{\proj}{{\mathscr{P}}}

\title{A Remark on stress of a spatially uniform dislocation density field}

\author{Siran Li}
\address{Siran Li: Department of Mathematics, Rice University, MS 136
P.O. Box 1892, Houston, Texas, 77251, USA.}

\email{\texttt{Siran.Li@rice.edu}}



\date{\today}

\begin{document}

\maketitle



\section{Introduction}
\subsection{}
In an interesting recent paper \cite{a}, Acharya proved that the stress produced by a spatially uniform dislocation density field in a body comprising a nonlinear elastic material may fail to vanish under no loads. The class of  counterexamples constructed in \cite{a} is essentially $2$-dimensional: it works with the subgroup $\mathcal{S}\OO(2) \oplus \langle{\bf Id}\rangle \subset \OO(3)$.  The objective of this note is to extend Acharya's result in \cite{a} to the $\OO(3)$, subject to one additional structural condition and less regularity assumptions.  

\subsection{Nomenclature}\label{notations} Throughout $\Omega \subset \R^3$ is a simply-connected bounded domain with outward unit normal vectorfield $\n$. The group of $3\times 3$  orthogonal matrices is denoted by $\OO(3)$; {\it i.e.}, $M \in \OO(3)$ if and only if $M^\top=M^{-1}$.  The special orthogonal group $\mathcal{S}\OO(2)$ consists of the matrices in $\OO(2)$ with determinant $1$. The matrix field $\FF: \Omega \map \gl$ designates the elastic distortion, and $\WW := \FF^{-1}$ whenever $\FF$ is invertible. $\TT: \gl \map  \OO(3)$ denotes a generally nonlinear, frame-indifferent stress response function, where $\gl$ is the space of $3\times 3$ matrices. The composition $\TT(\FF)$ is the symmetric Cauchy stress field applied to the configuration of body $\Omega$. The constant matrix $\aaa \in \gl$ denotes the dislocation density distribution specified on $\Omega$.

For a matrix field $M = \{M^i_j\}_{1 \leq i,j \leq m}:\Omega \map \gl$, its curl and divergence are understood in the \emph{row-wise} sense. In local coordinates it means the following: for each $i,j,k,\ell \in \{1,2,3\}$, $\curl M$ is the $2$-tensor field
\begin{align*}
\big[\curl M\big]^i_j := \na_k M^i_\ell - \na_\ell M^i_k
\end{align*}
where $(k,\ell,j)$ is an even permutation of $(1,2,3)$, and $\dv M$ is the vectorfield
\begin{equation*}
\big[\dv M\big]^i = \sum_j \na_j M^i_j. 
\end{equation*}

Moreover, recall the \emph{Leray projector} is the $L^2$-orthogonal projection $\proj:L^2(\R^3;\R^3) \map L^2(\R^3;\R^3)$ that sends a vectorfield in $\R^3$ onto its divergence-free part. On $\R^3$ it can be defined via Fourier transform:
\begin{equation*}
\widehat{\proj v}(\xi) := \bigg( {\bf Id} - \frac{\xi \otimes \xi}{|\xi|^2} \bigg)\hat{v}(\xi).
\end{equation*}
The Leray projector plays an important r\^{o}le in the mathematical analysis of incompressible Navier--Stokes equations; {\it cf. e.g.} Constantin--Foias \cite{cf} and Temam \cite{t}. For a matrix field $M$, $\proj(M)$ is again understood in the row-wise sense. We denote by $$\mathscr{Q}:={\bf Id}-\proj$$ the complementary projection of $\proj$.

\subsection{Differential Equations}
In the above setting, the governing equations for the internal stress field in the body subject to the Cauchy stress field $\TT(\FF)$ was derived by Willis in \cite{w}. See also Eq.~(3) in \cite{a}:
\begin{equation}\label{eq}
\begin{cases}
\curl \WW = - \aaa \qquad \text{ in } \Omega,\\
\dv\big(\TT(\FF)\big) = 0 \qquad \text{ in } \Omega,\\
\TT(\FF) \cdot \n = 0\qquad \text{ on } \p\Omega.
\end{cases}
\end{equation}
Here $\aaa$ is a prescribed constant matrix. This PDE system is considered under the following
\begin{assumption}\label{assumption}
$\TT(\FF)={\bf 0}$ if and only if $\FF$ takes values in $\OO(3)$.
\end{assumption}
Acharya proved in \cite{a} the following result: 
\begin{theorem}\label{thm: a}
Let $\Omega$, $\WW$, $\TT$, $\FF$, and $\n$ be as in Section~\ref{notations} above. Let $\alpha$ be any nonzero constant matrix. Then, under Assumption~\ref{assumption}, there does not exist $\theta \in C^2(\Omega;\R)$ such that $\WW=\rt$ is a solution for Eq.~\eqref{eq}; here 
\begin{equation}\label{Rtheta}
\rt:=\begin{bmatrix}
\cos\theta&-\sin\theta&0\\
\sin\theta&\cos\theta&0\\
0&0&1
\end{bmatrix}.
\end{equation}
\end{theorem}
The proof in \cite{a} follows from concrete computations: with the ansatz~\eqref{Rtheta}, Eq.~\eqref{eq} reduces to a system of algebraic equations for $\sin\theta$ and $\cos\theta$ only, which is not soluble unless $\aaa \equiv 0$.

The goal of this note is to extend Acharya's Theorem~\ref{thm: a} in order to include more general form of $\WW$ and assuming lower regularity requirements. At the moment we are not able to generalise to all of $\OO(3)$-valued $\WW$; an additional structural condition is needed ---
\begin{assumption}\label{assumption'}
$\mathscr{Q}(\WW)$ is $\OO(3)$-valued ($\mathscr{Q}$ is 
the complement of Leray projector in Section~\ref{notations}).
\end{assumption}

\subsection{Mechanics}
In the terminologies of continuum mechanics, Theorem~\ref{thm: a} means that in the \emph{nonlinear} regime, there is no $C^2$-stress-free spatially uniform dislocation density field, unless such uniform dislocation density is everywhere vanishing. 

Various dislocation distributions producing no stress have been observed in the limit of continuum elastic descriptions ({\it cf.} Mura \cite{m}, Head--Howison--Ockendon--Tighe  \cite{hhot}, Yavari--Goriely \cite{yg}, etc.). This is the background for our work. In this note, we aim to further the investigation by Acharya \cite{a} in the nonlinear regime.

\section{Main Result}
\begin{theorem}\label{thm: main}
Let $\Omega$, $\WW$, $\aaa$, $\TT$, $\FF$, and $\n$ be as in Section~\ref{notations}. Under Assumptions~\ref{assumption} and \ref{assumption'}, Eq.~\eqref{eq} has no solution $\WW$ in $C^1(\Omega; \OO(3))$ unless the uniform dislocation density field $\alpha \equiv 0$.
\end{theorem}


Theorem~\ref{thm: main} agrees with the linear case. The following arguments are essentially taken from Section~3 in \cite{a}. When $\mathbf{U}:=\FF-{\bf Id}$ is uniformly small, set $\mathbf{C}:=D\TT({\bf I})$. The matrix field ${\bf U}$ is known as the elastic distortion, and the rank-$4$ tensor field $\mathbf{C}$ is known as the elastic modulus. Then the \emph{linearised system} for Eq.~\eqref{eq} is 
\begin{equation}\label{ linearised eq}
\begin{cases}
\curl \mathbf{U} = - \aaa \qquad \text{ in } \Omega,\\
\dv\big(\mathbf{C}\mathbf{U}\big) = 0 \qquad \text{ in } \Omega,\\
\mathbf{C}\mathbf{U} \cdot \n = 0\qquad \text{ on } \p\Omega.
\end{cases}
\end{equation}
By Kirchhoff's uniqueness theorem for linear elastostatics, the symmetric part $${\bf \e}:= \frac{{\bf U} + {\bf U}^\top}{2}$$
must be zero. Thus, if ${\bf U}$ is $\OO(3)$-valued, then Eq.~\eqref{ linearised eq} is not soluble except when $\aaa \equiv 0$. That is, $\alpha \equiv 0$ is a necessary (in fact, not sufficient in general) condition for the solubility of Eq.~\eqref{ linearised eq}. 

Also note that $\WW=\mathcal{R}_\theta$ in Theorem~\ref{thm: a} satisfies Assumption~\ref{assumption'}: direct computation in polar coordinates shows that ${\rm div}\,\mathcal{R}_\theta \equiv 0$; hence $\mathscr{Q}\WW\equiv\WW\equiv \mathcal{R}_\theta$, which is $\OO(3)$-valued.

\section{Proof}
\begin{proof}[Proof of Theorem~\ref{thm: main}]

Throughout the proof we denote by $\WW^1, \WW^2, \WW^3$ the row-vectorfields of the matrix field $\WW$. Also, let $\ta$ be the field of differential $2$-forms dual to $\aaa$, namely
\begin{align*}
\ta^i = \aaa^i_1 \,dx^2 \wedge dx^3 + \aaa^i_2 \,dx^3 \wedge dx^1 + \aaa^i_3 \,dx^1 \wedge dx^2.
\end{align*}
Thus, by Hodge duality, the first equation in Eq.~\eqref{eq} becomes
\begin{equation}\label{curl eq}
d \WW^i = -\ta^i\qquad \text{ for each } i \in \{1,2,3\},
\end{equation}
which is an identity of  $2$-forms. Here and hereafter, we identify $\WW^i$ with a $1$-form (not relabelled).

Under Assumption~\ref{assumption} the second and the third equations in Eq.~\eqref{eq} are satisfied automatically. So it remains to solve for Eq.~\eqref{curl eq} in the space of $\OO(3)$-valued matrix fields.

 Recall that the divergence operator acting on differential $1$-forms on $\Omega \subset \R^3$ is nothing but the codifferential $d^* := \star d \star$, where $\star$ is the Hodge star operator. Also, the Laplacian equals
\begin{align}\label{laplace-beltrami}
\Delta = dd^*+d^*d.
\end{align}
Let us split $\WW$ into
\begin{equation}\label{hodge}
\WW^i = d^*\two^i + d\phi^i + c^i\qquad \text{ on } \Omega,
\end{equation}
where $\two^i$ is a field of differential $2$-form, $\phi^i$ is a scalarfield, and $c^i$ is a constant in $\R^3$. This is done by the Hodge decomposition theorem and that $\Omega$ is simply-connected; see, {\it e.g.}, Chapter 6 in \cite{wa}. In local coordinates, Eq.~\eqref{hodge} can be expressed as follows:
\begin{align*}
\WW^i_j = \sum_{k=1}^3\na_k \two^i_{kj} + \na_j \phi^i + c^i_j\qquad \text{ for each } i,j \in \{1,2,3\}.
\end{align*}
By standard elliptic regularity theory (see \cite{gt}), $\two^i$ and $\phi^i$ have $C^{1,\gamma}$-regularity for any $\gamma \in [0,1[$. 

Now we \emph{claim} that
\begin{equation}\label{claim}
\Big\{\na_j\phi^i\Big\}_{1\leq i,j \leq 3} \text{ is equal to a constant $\OO(3)$-matrix}.
\end{equation}
Indeed, since the Leray projector maps onto the divergence-free part of $\WW$, we have $\mathscr{Q}\WW^i = d\phi^i$ for $\phi^i \in C^{1,\gamma}(\Omega)$. By Assumption~\ref{assumption'} we have
\begin{equation*}
\sum_{k=1}^3 \na_k \phi^i \na_k \phi^j = \delta^{ij},
\end{equation*}
namely that $\phi$ is an isometric embedding from $\Omega \subset \R^3$ into $\R^3$. The classical rigidity theorem of Liouville~\cite{l} yields that $\phi^i$ is an affine map globally on $\Omega$ (in fact, $C^1$-regularity of $\phi^i$ suffices here). Thus the \emph{claim}~\eqref{claim} follows.

To conclude the proof, taking $d^*$ to both sides of Eq.~\eqref{hodge} and noting the \emph{claim}~\eqref{claim}, we get
\begin{equation*}
d^*\WW^i = 0.
\end{equation*}
This together with Eqs.~\eqref{curl eq} and \eqref{laplace-beltrami} implies that
\begin{align}\label{harmonic}
\Delta \WW^i=0.
\end{align}
That is, $\WW^i$ is a harmonic $1$-form for each $i\in\{1,2,3\}$. Eq.~\eqref{harmonic} is understood in the sense of distributions; nevertheless, by Weyl's lemma (see \cite{gt}) $W^i$ is automatically $C^\infty$. In view again of the Hodge theory (see Chapter 6 in \cite{w}), it is represented by generators of the first cohomology group. But $\Omega$ is simply-connected, so there is no non-trivial such generator. Thus $\WW^i$ is  constant. Therefore, we infer from Eq.~\eqref{curl eq} that $\alpha^i$ equals zero. The proof is complete. \end{proof}

\section{Remarks}

It would be interesting to consider the same problem for $\Omega$ being a $3$-dimensional manifold, which falls into the framework of incompatible (non-Euclidean) elasticity. 

The mechanical problem considered in this paper may have deep underlying geometrical connotations. In particular, it is related to  constructions for coframes with prescribed (closed) differential. See Bryant--Clelland \cite{bc} for analyses via exterior differential systems.

\bigskip
\noindent
{\bf Acknowledgement}. The author is deeply indebted to Amit Acharya for kind communications and insightful discussions. We also thank Janusz Ginster for pointing out a fallible argument in an  earlier version of the draft.


\begin{thebibliography}{99}

\bibitem{a}
A. Acharya, {Stress of a spatially uniform dislocation density field},  \textit{J. Elasticity} \textbf{137} (2019), no. 2, 151--155 

\bibitem{bc}
R.~L. Bryant and J.~N. Clelland, {Flat metrics with a prescribed derived coframing}, \textit{SIGMA Symmetry Integrability Geom. Methods Appl.} \textbf{16} (2020), Paper No. 004, 23 pp

\bibitem{cf}
P. Constantin and C. Foias, \textit{Navier--Stokes Equations}, University of Chicago Press, 1988

\bibitem{hhot}
A.~K. Head, S.~D. Howison, J.~R. Ockendon, and S.~P. Tighe, {An equilibrium theory of dislocation continua}, \textit{SIAM Rev.} \textbf{35} (1993), 580--609

\bibitem{l}
J. Liouville, {Th\'{e}or\`{e}me sur l'\'{e}quation $dx^2 + dy^2 + dz^2 = \lambda d\alpha^2 + d\beta^2+d\gamma^2$}, \textit{J.
Math. Pures Appl.} (1850).

\bibitem{m}
T. Mura, {Impotent dislocation walls}, \textit{Materials Science and Engineering: A} \textbf{113} (1989), 149--152


\bibitem{t}
R. Temam, \textit{Navier–Stokes Equations: Theory and Numerical Analysis}, AMS Chelsea Publishing, 2001


\bibitem{gt}
D. Gilbarg and N.~S. Trudinger, \textit{Elliptic partial differential equations of second order}, Classics in Mathematics. Springer-Verlag, Berlin, 2001

\bibitem{wa}
F.~W. Warner, \textit{Foundations of differentiable manifolds and Lie groups}. Corrected reprint of the 1971 edition. Graduate Texts in Mathematics, 94. Springer-Verlag, New York-Berlin, 1983

\bibitem{w}
J.~R. Willis, {Second-order effects of dislocations in anisotropic crystals}, \textit{Intern. J. Engineering Sci.} \textbf{5} (1967), 171--190

\bibitem{yg}
A. Yavari and A. Goriely, {Riemann--Cartan geometry of nonlinear dislocation mechanics}, \textit{Arch. Ration. Mech. Anal.} \textbf{205} (2012), 59--118

\end{thebibliography}
\end{document}